\documentclass[reprint,aps,prl,twocolumn]{revtex4-1}

\usepackage{times}

\usepackage{amsmath}
\usepackage[titletoc,title]{appendix}
\usepackage{graphicx,epic,eepic,epsfig,amsmath,latexsym,amssymb,verbatim,color}
\usepackage{dsfont}
\usepackage{float}
\usepackage{cases}
\usepackage{tikz}
\usepackage{hyperref}
\hypersetup{colorlinks=true,citecolor=blue,linkcolor=blue,filecolor=blue,urlcolor=blue,breaklinks=true}

\makeatletter
\def\grd@save@target#1{%
  \def\grd@target{#1}}
\def\grd@save@start#1{%
  \def\grd@start{#1}}
\tikzset{
  grid with coordinates/.style={
    to path={%
      \pgfextra{%
        \edef\grd@@target{(\tikztotarget)}%
        \tikz@scan@one@point\grd@save@target\grd@@target\relax
        \edef\grd@@start{(\tikztostart)}%
        \tikz@scan@one@point\grd@save@start\grd@@start\relax
        \draw[minor help lines,magenta] (\tikztostart) grid (\tikztotarget);
        \draw[major help lines] (\tikztostart) grid (\tikztotarget);
        \grd@start
        \pgfmathsetmacro{\grd@xa}{\the\pgf@x/1cm}
        \pgfmathsetmacro{\grd@ya}{\the\pgf@y/1cm}
        \grd@target
        \pgfmathsetmacro{\grd@xb}{\the\pgf@x/1cm}
        \pgfmathsetmacro{\grd@yb}{\the\pgf@y/1cm}
        \pgfmathsetmacro{\grd@xc}{\grd@xa + \pgfkeysvalueof{/tikz/grid with coordinates/major step}}
        \pgfmathsetmacro{\grd@yc}{\grd@ya + \pgfkeysvalueof{/tikz/grid with coordinates/major step}}
        \foreach \x in {\grd@xa,\grd@xc,...,\grd@xb}
        \node[anchor=north] at (\x,\grd@ya) {\pgfmathprintnumber{\x}};
        \foreach \y in {\grd@ya,\grd@yc,...,\grd@yb}
        \node[anchor=east] at (\grd@xa,\y) {\pgfmathprintnumber{\y}};
      }
    }
  },
  minor help lines/.style={
    help lines,
    step=\pgfkeysvalueof{/tikz/grid with coordinates/minor step}
  },
  major help lines/.style={
    help lines,
    line width=\pgfkeysvalueof{/tikz/grid with coordinates/major line width},
    step=\pgfkeysvalueof{/tikz/grid with coordinates/major step}
  },
  grid with coordinates/.cd,
  minor step/.initial=.2,
  major step/.initial=1,
  major line width/.initial=2pt,
}
\makeatother

\usepackage[marginal]{footmisc}
\usepackage{url}
\usepackage{theorem}
\newtheorem{definition}{Definition}

\newtheorem{lemma}[definition]{Lemma}

\newtheorem{theorem}[definition]{Theorem}

\def\squareforqed{\hbox{\rlap{$\sqcap$}$\sqcup$}}
\def\qed{\ifmmode\squareforqed\else{\unskip\nobreak\hfil
\penalty50\hskip1em\null\nobreak\hfil\squareforqed
\parfillskip=0pt\finalhyphendemerits=0\endgraf}\fi}
\def\endenv{\ifmmode\;\else{\unskip\nobreak\hfil
\penalty50\hskip1em\null\nobreak\hfil\;
\parfillskip=0pt\finalhyphendemerits=0\endgraf}\fi}

\newenvironment{proof}{\noindent \textbf{{Proof~}}}{\hfill $\blacksquare$}
\newenvironment{remark}{\noindent \textbf{{Remark~}}}{}

\mathchardef\ordinarycolon\mathcode`\:
\mathcode`\:=\string"8000
\def\vcentcolon{\mathrel{\mathop\ordinarycolon}}
\begingroup \catcode`\:=\active
  \lowercase{\endgroup
  \let :\vcentcolon
  }

\makeatletter
\def\resetMathstrut@{%
	\setbox\z@\hbox{%
		\mathchardef\@tempa\mathcode`\[\relax
		\def\@tempb##1"##2##3{\the\textfont"##3\char"}%
		\expandafter\@tempb\meaning\@tempa \relax
	}%
	\ht\Mathstrutbox@\ht\z@ \dp\Mathstrutbox@\dp\z@}
\makeatother
\begingroup
\catcode`(\active \xdef({\left\string(}
\catcode`)\active \xdef){\right\string)}
\endgroup
\mathcode`(="8000 \mathcode`)="8000

\newcommand{\nc}{\newcommand}
\nc{\rnc}{\renewcommand}
\nc{\lbar}[1]{\overline{#1}}
\nc{\bra}[1]{\langle#1|}
\nc{\ket}[1]{|#1\rangle}
\nc{\ketbra}[2]{|#1\rangle\!\langle#2|}
\nc{\braket}[2]{\langle#1|#2\rangle}

\nc{\proj}[1]{| #1\rangle\!\langle #1 |}
\nc{\avg}[1]{\langle#1\rangle}
\nc{\Rank}{\operatorname{Rank}}
\nc{\smfrac}[2]{\mbox{$\frac{#1}{#2}$}}
\nc{\tr}{\operatorname{Tr}}
\nc{\ox}{\otimes}
\nc{\dg}{\dagger}
\nc{\dn}{\downarrow}
\nc{\cA}{\boldsymbol{\cal A}}
\nc{\ca}{\boldsymbol{a}}
\nc{\cB}{{\cal B}}
\nc{\cC}{{\cal C}}
\nc{\cD}{{\cal D}}
\nc{\cE}{{\cal E}}
\nc{\cF}{{\cal F}}
\nc{\cG}{{\cal G}}
\nc{\cH}{{\cal H}}
\nc{\cI}{{\cal I}}
\nc{\cJ}{{\cal J}}
\nc{\cK}{{\cal K}}
\nc{\cL}{{\cal L}}
\nc{\cM}{\boldsymbol{\cal M}}
\nc{\cN}{{\cal N}}
\nc{\cO}{{\cal O}}
\nc{\cP}{{\cal P}}
\nc{\cQ}{\boldsymbol{\cal Q}}
\nc{\cq}{\boldsymbol{q}}
\nc{\cR}{{\cal R}}
\nc{\cS}{{\cal S}}
\nc{\cT}{{\cal T}}
\nc{\cV}{{\cal V}}
\nc{\cX}{{\cal X}}
\nc{\cx}{\boldsymbol{x}}
\nc{\cY}{{\cal Y}}
\nc{\cZ}{{\cal Z}}
\nc{\cW}{{\cal W}}
\nc{\csupp}{{\operatorname{csupp}}}
\nc{\qsupp}{{\operatorname{qsupp}}}
\nc{\var}{{\operatorname{var}}}
\nc{\rar}{\rightarrow}
\nc{\lrar}{\longrightarrow}
\nc{\polylog}{{\operatorname{polylog}}}
\nc{\1}{{\mathds{1}}}
\nc{\wt}{{\operatorname{wt}}}
\nc{\av}[1]{{\left\langle {#1} \right\rangle}}
\nc{\supp}{{\operatorname{supp}}}

\def\x{\xi}

\def\o{\omega}

\def\O{\Omega}

\nc{\RR}{{{\mathbb R}}}
\nc{\CC}{{{\mathbb C}}}
\nc{\FF}{{{\mathbb F}}}
\nc{\NN}{{{\mathbb N}}}
\nc{\ZZ}{{{\mathbb Z}}}
\nc{\PP}{{{\mathbb P}}}
\nc{\QQ}{{{\mathbb Q}}}
\nc{\UU}{{{\mathbb U}}}
\nc{\EE}{{{\mathbb E}}}
\nc{\id}{{\operatorname{id}}}

\nc{\<}{\langle}
\rnc{\>}{\rangle}
\nc{\Hom}[2]{\mbox{Hom}(\CC^{#1},\CC^{#2})}
\nc{\rU}{\mbox{U}}

\nc{\ob}[1]{#1}

\nc{\SEP}{{\text{SEP}}}
\nc{\NS}{{\text{NS}}}
\nc{\LOCC}{{\text{LOCC}}}
\nc{\PPT}{{\text{PPT}}}
\nc{\EXT}{{\text{EXT}}}
\nc{\Sym}{{\operatorname{Sym}}}

\definecolor{darkblue}{RGB}{0,76,156}
\definecolor{darkkblue}{RGB}{0,0,153}
\definecolor{blue2}{RGB}{102,178,255}
\definecolor{myred}{RGB}{180,5,4}

\RequirePackage[framemethod=default]{mdframed}
\newmdenv[skipabove=7pt,
skipbelow=7pt,
innerleftmargin=5pt,
innerrightmargin=5pt,
innertopmargin=5pt,
leftmargin=0cm,
rightmargin=0cm,
innerbottommargin=5pt,
linewidth=1pt]{tBox}

\newmdenv[skipabove=7pt,
skipbelow=7pt,
backgroundcolor=blue2!25,
innerleftmargin=5pt,
innerrightmargin=5pt,
innertopmargin=5pt,
leftmargin=0cm,
rightmargin=0cm,
innerbottommargin=5pt,
linewidth=1pt]{dBox}

\newmdenv[skipabove=7pt,
skipbelow=7pt,
backgroundcolor=darkkblue!15,
innerleftmargin=5pt,
innerrightmargin=5pt,
innertopmargin=5pt,
leftmargin=0cm,
rightmargin=0cm,
innerbottommargin=5pt,
linewidth=1pt]{sBox}

\usepackage{enumitem}

\newtheorem{innercustomgeneric}{\customgenericname}
\providecommand{\customgenericname}{}
\newcommand{\newcustomtheorem}[2]{%
  \newenvironment{#1}[1]
  {%
   \renewcommand\customgenericname{#2}%
   \renewcommand\theinnercustomgeneric{##1}%
   \innercustomgeneric
  }
  {\endinnercustomgeneric}
}

\newcustomtheorem{customthm}{Theorem}
\newcustomtheorem{customlemma}{Lemma}

\newcommand{\sspan}{\text{Span}}
\newcommand{\MC}{\text{\rm HC}}

\begin{document}
\title{Quantum Advantages in Hypercube Game}

\author{Xiaoyu He$^{1}$}
\email{hexiaoyu14@mails.ucas.ac.cn}

\author{Kun Fang$^{3}$}
\email{kun.fang-1@student.uts.edu.au}

\author{Xiaoming Sun$^{1}$}
\email{sunxiaoming@ict.ac.cn}

\author{Runyao Duan$^{2,3}$}
\email{duanrunyao@baidu.com}

\affiliation{$^1$Institute of Computing Technology, Chinese Academy of Sciences}

\affiliation{$^2$Institute for Quantum Computing, Baidu Inc., Beijing 100193, China}

\affiliation{$^3$Centre for Quantum Software and Information, Faculty of Engineering and Information Technology,\\ University of Technology Sydney, NSW 2007, Australia}

\begin{abstract}
We introduce a novel generalization of the Clauser-Horne-Shimony-Holt (CHSH) game to a multiplayer setting, i.e., Hypercube game, where all $m$ players are required to assign values to vertices on corresponding facets of an $m$-dimensional hypercube. The players win if and only if their answers satisfy both parity and consistency conditions. We completely characterize the maximum winning probabilities (game value) under classical, quantum and no-signalling strategies, respectively. In contrast to the original CHSH game designed to demonstrate the superiority of quantumness, we find that the quantum advantages in the Hypercube game significantly decrease as the number of players increase. Notably, the quantum value decays exponentially fast to the classical value as $m$ increases, while the no-signalling value always remains to be one. 
\end{abstract}
\maketitle

\textit{Introduction.}---~Nonlocal games are played by a number of cooperating players against a referee. The players are required to reply to the referee's randomly selected questions with appropriate answers to win the game. Their goal is to collaborate and maximize the average winning probability. Before the game starts, players may agree upon a common strategy. Then they move far apart and cannot communicate with each other while the game is being played. Communications are only allowed between players and the referee.

The nonlocal game model plays a crucial role in both complexity theory and theoretical physics. In the complexity theory, it is closely related with the interactive proof system~\cite{ben1988multi}, efficient proof verification~\cite{babai1991non}, hardness of approximation~\cite{feige1996interactive} and the PCP conjecture~\cite{arora1998proof,arora1998probabilistic}. In theoretical physics, a well-known consequence of earlier works~\cite{bell1964instein,kochen1967problem,clauser1969proposed} shows that quantum entanglement shared between players can allow them to outperform all purely classical strategies for some nonlocal games. This phenomenon confirmed experimentally~\cite{Rosenfeld2017} turns out to be useful in practice, particularly for the self-testing of quantum states~\cite{mayers2004self,coladangelo2017all,vsupic2017simple}, randomness generation~\cite{pironio2010random} and secure key distribution~\cite{acin2007device}.

One of the most important game is the so-called  Clauser-Horne-Shimony-Holt (CHSH) game~\cite{clauser1969proposed,cleve2004consequences} which can be equivalently~\footnote{See the Supplemental Material below.} expressed as follows. Consider a square with two rows and two columns shown in Fig.~\ref{CHSH game} where we use $A_0$, $A_1$, $B_0$, $B_1$ to mark the four edges respectively. The referee randomly chooses two bits $q_1,q_2\in\{0,1\}$ according a uniform distribution and sends $q_1$ to Alice, $q_2$ to Bob. They need to respond to the referee their assignments of two values from $\{+1,-1\}$, one for each vertex on the $q_1$-th column (marked with $A_{q_1}$) and $q_2$-th row (marked with $B_{q_2}$) respectively. They win if and only if their assignments are consistent on their common vertex and the product of their own assignments equals to $1$ except that Alice's product equals to $-1$ if $q_1=1$. 
\vspace{-0.3cm}
\begin{figure}[H]
\begin{minipage}[b]{0.25\textwidth}
\centering
\includegraphics[scale = 1]{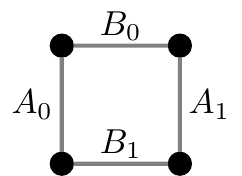}
\end{minipage}\hspace{-0.6cm}
\begin{minipage}[b]{0.22\textwidth}
\caption{\small CHSH game or\\ Hypercube game with $m=2$.}
\label{CHSH game}
\end{minipage}
\end{figure}
\vspace{-0.2cm}
\noindent It is known that quantum entanglement can improve the average winning probability of the CHSH game to approximately $0.85$ while the best classical strategy can only win the game with probability $0.75$.

While nonlocal game has been extensively considered in the bipartite scenario (e.g.~\cite{clauser1969proposed,cleve2004consequences,mermin1990simple,peres1990incompatible,regev2013quantum,lupini2018perfect}), the general multipartite case is less explored. In this work, we aim to introduce a novel generalization of the CHSH game to a multiplayer game, called $m$-player Hypercube game ($\MC_m$). We determine its classical, quantum and no-signalling game values by explicitly constructing the game strategies and showing their optimality. In particular, we find that the quantum advantage decreases as more players involved in this game. Specifically, the quantum value exponentially decays to the classical value as $m$ increasing, while the no-signalling value always remains to be one. Compared with the ``Guess your neighbor's input'' game~\cite{almeida2010guess} which shows no quantum advantage at all, the quantum advantage in $\MC_m$ always exists in spite of the fact that it vanishes asymptotically.

\vspace{0.1cm}
\textit{Nonlocal game and strategies.}---~Consider a general nonlocal game with $m$ players. Let $\cQ_i$ and $\cA_i$ denote the finite sets of possible questions and answers for the $i$-th player $\cP_i$, respectively. Denote $\cQ:=(\cQ_1,\cdots,\cQ_m)$ and $\cA:=(\cA_1,\cdots,\cA_m)$. The referee initiates the game by selecting a question $\cq:= (q_1,\cdots,q_m) \in \cQ$ according to a probability distribution $\pi: \cQ \to [0,1]$ and send each question $q_i$ to the $i$-th player. Given their questions, the players are required to provide the referee with an answer $\ca:=(a_1,\cdots,a_m) \in \cA$ where $a_i$ denotes the answer by the $i$-th player respectively. The probability of answering $\ca$ conditioned on the given question $\cq$ is denoted as $P(\ca|\cq)$, which is referred to as a \emph{correlation function}. Each correlation function corresponds to a strategy taken by the players. Finally the referee evaluates some predicate $V:\cQ\times \cA\to \{0,1\}$ to determine whether the players win, $1$ or lose, $0$. We denote the game as $G=(\pi,V)$.
Then the maximum probability with which players can win the game $G$, is defined as
\begin{align}\label{game values}
  \o_\O(G) := \sup_{P\in \O} \ \sum_{\cq \in \cQ}\sum_{\ca \in \cA} \pi(\cq)  P(\ca|\cq) V(\ca|\cq),
\end{align}
where the supreme is taken over all possible correlation functions in a certain class $\O$. 

A \textit{classical strategy} consists of a function for each player $f_i: \cQ_i \to \cA_i$ that deterministically produces an answer for every question. This type of classical strategy is referred to as a \textit{deterministic strategy}. It corresponds to a \textit{deterministic correlation function} defined by
\begin{align}
  P(\ca|\cq)=\begin{cases} 1,\quad a_i = f_i(q_i), \forall\, i\\ 0, \quad \text{otherwise}. \end{cases}
\end{align}
Since the \textit{stochastic strategy} via shared randomness does not provide advantage to achieve higher winning probability, there is no loss of generality to restrict our attention only to the deterministic ones. The \textit{classical value} denoted as $\o_c(G)$ is the supremum of Eq.~\eqref{game values} taking over all the possible deterministic correlation functions. Note that approximating the classical value of a game is NP-hard in general~\cite{fortnow1994power,kempe2011entangled,ito2009oracularization}.

In a \textit{quantum strategy}, the players may prepare and share a joint quantum state $\ket{\psi}$ prior to the game. Then they can perform a local (projective) measurement $\cM_{q_i,i}:=\{M_{q_i,i}^{a_i}\}_{q_i}$ respectively on their subsystems dependent on their received questions $q_i$ and respond the answers $a_i$ according to their measurement outcomes. The subscript $i$ indicates that the measurement is acting on the $i$-th subsystem. This strategy corresponds to a \textit{quantum correlation function} defined as
\begin{align}
  P(\ca|\cq) = \<M_{q_1,1}^{a_1}\ox \cdots \ox M_{q_m,m}^{a_m}\>_\psi,
\end{align}
where $\<M\>_\psi := \tr M\psi$ and $\psi=\ket{\psi}\bra{\psi}$. Sometimes we ignore the symbol of tensor product for simplicity. The \textit{quantum value} of the game $G$, denoted as $\o_q(G)$, is the supremum of Eq.~\eqref{game values} taking over all the possible quantum correlation functions. The quantum value of a game is QMA-hard to approximate in general~\cite{ji2016compression,fitzsimons2015multiprover,ji2016classical} and there is no standard method to compute it.
Note that strategies involving mixed states or general measurements can all be represented in the above form by expanding the dimension of the initial state, due to the Naimark's theorem~\cite{naimark1943representation}.  

\textit{No-signalling strategy} is a more general type of strategy, which has the corresponding correlation function satisfying \textit{no-signalling correlations}. That is, for any index subset $I=\{i_1,\cdots,i_k\} \subseteq \{1,\cdots,m\}$, the marginal probability satisfies
\begin{align}
  P(\ca_I|\cq) = P(\ca_I|\cq_I),
\end{align}
with $\ca_I:=(a_{i_1},\cdots,a_{i_k})$ and $\cq_I:=(q_{i_1},\cdots,q_{i_k})$. This guarantees that any subset of the players cannot signal their received questions to the others. In general, it suffices to consider the subset $I$ with cardinality one. The \textit{no-signalling value} of the game $G$, denoted as $\o_{ns}(G)$, is the supremum of Eq.~\eqref{game values} taking over all the possible no-signalling correlation functions. Since the no-signalling correlation is characterized by linear constraints, the no-signalling value of a game can be computed via  linear programming.

\vspace{0.2cm}
\textit{$m$-player Hypercube game.}---~Inspired by the geometric approach to defining CHSH game in the Introduction, we generalize this game to the multipartite case where each player is required to assign values to vertices on some facet of an $m$-dimensional hypercube. Let us consider an $m$-dimensional hypercube with $2^m$ vertices and introduce the $m$-player Hypercube game as follows.

\begin{tBox}
\textbf{\underline{Hypercube game ($\MC_m$)}}
\hspace{-0.5cm}
\begin{enumerate}[leftmargin=0.5cm]
\item The referee randomly chooses one of the question $\cq=(q_1,\dots,q_m)\in \{0,1\}^m$ according to a uniform distribution, and sends $q_i$ to the player $\cP_i$. 
\item Each player $\cP_i$ needs to assign $1$ or $-1$ to $2^{m-1}$ vertices in the set
$\cx_i:= \left\{\,\cx \in\{0,1\}^m | \ x_i=q_i\,\right\},$
and sends their assignments to the referee.

\item The players win if and only if 
\begin{enumerate}
  \item (\textit{Parity}): the product of their own assignments equals to $1$ except that the product of the first player's assignments equals to $-1$ if $q_1=1$;
  \item (\textit{Consistency}): the assignments are consistent on all the common vertices $\cx_i \cap \cx_j, \forall\, i\neq j$. 
\end{enumerate}
\end{enumerate}
\end{tBox}

To provide some intuition, we illustrate the case $m=3$ in Fig.~\ref{MC3}. If the selected question is $\cq = (0,0,0)$, each player needs to assign values to the vertices on their marked facet $A_{0}$, $B_0$ and $C_0$ respectively. They win if and only if the products of values on $A_0$, $B_0$ and $C_0$ are all equal to $1$ and their common vertices (marked with red hollow circle) are assigned the same values from different players.

\begin{figure}[H]
\centering
\includegraphics[scale=0.9]{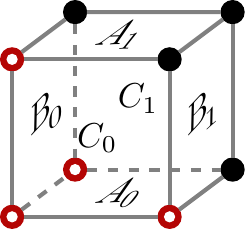}
\caption{Hypercube game with $m=3$. The facets are marked as: bottom ($A_0$), top ($A_1$), left ($B_0$), right ($B_1$), front ($C_0$), back ($C_1$).}
\label{MC3}
\end{figure}

The more players involved, the more difficult for the players to cooperate and win the game. We confirm this intuition and show the following result.
\begin{theorem}\label{MCm theorem}
  For the $m$-player Hypercube game $(\MC_m)$, its classical and no-singalling values are respectively given by
  \begin{align}
  \o_c(\MC_m)=\frac12+\frac{1}{2^m}\quad \text{and}
  \quad \o_{ns}(\MC_m) = 1.
  \end{align}
   Its quantum value $\o_{q}(\MC_m)$ is given by the single letter optimization as follows
\begin{align}\label{quantum value}
\frac{1}{2^m}\max_\theta \big((1+\cos\theta)^{m-1} + (1+\sin\theta)^{m-1}\big),
\end{align}
which can be approximated as 
\begin{align}
\o_q(\MC_m) = \frac{1}{2}+\frac{1}{2^m}+O(\frac{m-1}{4^{m}}).
\end{align}
\end{theorem}

For any finite number of players, there always exists a quantum strategy outperforming any classical ones. However, this quantum advantage becomes negligible as the number of players becomes larger and both quantum and classical value exponentially decay to $\frac12$. Surprisingly, the no-signalling correlation is strong enough to resist such decay as shown in Fig.~\ref{compare fig}. We show the proof of Theorem~\ref{MCm theorem} in the following sections.

\begin{figure}[H]
  \centering
  \includegraphics{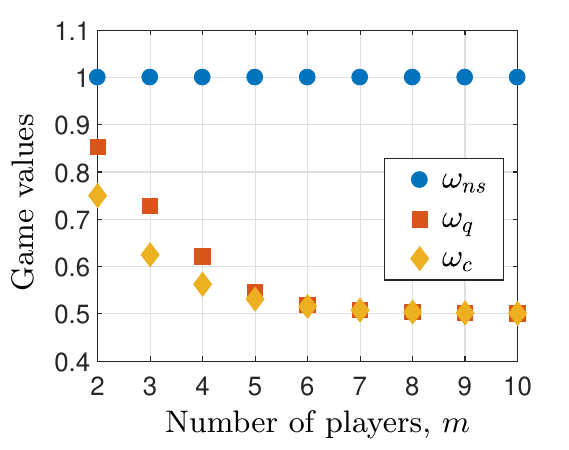}
  \caption{\small This figure depicts the classical (yellow diamond), quantum (red square), and no-signalling (blue dot) values of $\MC_m$.}
  \label{compare fig}
\end{figure}

\vspace{0.1cm}
\textit{Hypercube game with quantum strategies.}---~In this section, we outline the idea to show the quantum value of $\MC_m$. Our game strategy is as follows. The assignments to the four vertices $\cx_{0,0}:=(0,0,\dots,0)$, $\cx_{0,1}:=(0,1,\dots,1)$, $\cx_{1,0}:=(1,0,\dots,0)$ and $\cx_{1,1}:=(1,1,\dots,1)$ depend on the measurement outcomes while others are assigned to be $1$. This reduces the $\MC_m$ to a game similar to the CHSH game. That is, the first player ($\cP_1$) acts as Alice and assigns the vertices on the columns while other players ($\cP_i,i\geq 2$) play the same role as Bob and assign the vertices on the rows, as shown in Fig.~\ref{Reduced MCm}. Since all the common vertices between every two players ($\cP_i,\cP_j, i\neq j$) are the ones intersecting with the first player, we only need to check the consistency between $\cP_1$ and $\cP_i$ for $i \geq 2$, i.e., 
\begin{align}\label{consis condition}
   a_1(q_1,q_i,\dots,q_i)=a_i(q_1,q_i,\dots,q_i), \forall\, i \geq 2,
\end{align}
where $a_i(\cx)$ denotes the value that the $i$-th player assigns to the vertex $\cx$.

\begin{figure}[H]
\centering
\includegraphics[scale=0.9]{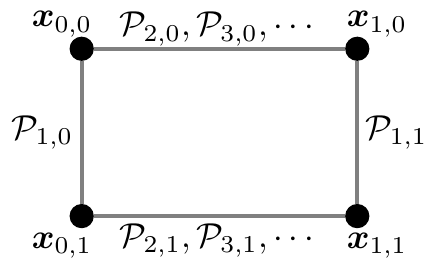}
\caption{\small Each player ($\cP_i$) needs to assign vertices on the line marked with $\cP_{i,q_i}$ for $q_i\in \{0,1\}$.}
\label{Reduced MCm}
\end{figure}

In this quantum strategy, each player holds one qubit of a $m$-qubit GHZ state
\begin{align}
\ket{\psi} = (\ket{0}^{\otimes m} + \ket{1}^{\otimes m})/\sqrt{2},
\end{align}
and performs projective measurement associated with the observable $\cO_{q_i,i} = Z_{\theta_{q_i,i}}$ where 
\begin{align*}
Z_\theta :=
\begin{bmatrix}
 \cos\theta & \sin\theta \\
 \sin\theta & -\cos\theta
\end{bmatrix} \ \ \text{and}\ \
  \theta_{q_i,i}= \begin{cases} 
     \ \ q_i\cdot \frac{\pi}{2} & \text{if}\ i = 1,\\[2pt]
     (-1)^{q_i}\alpha & \text{if}\ i \geq 2.
  \end{cases}
\end{align*}
Note that this is a generalization of the optimal quantum strategy for the CHSH game.
Denote $o_{q_i,i}$ as the $i$-th player's measurement outcome given the question $q_i$. Then the players can perform the assignments as follows:
\begin{align}
  a_1(q_1,0,\cdots,0) & = o_{q_1,1},\\
  a_1(q_1,1,\cdots,1) & = (-1)^{q_1}o_{q_1,1},\\
   a_i(0,q_i,\dots,q_i) & = o_{q_i,i},\ \forall\, i \geq 2,\\
  a_i(1,q_i,\dots,q_i) & = o_{q_i,i},\ \forall\, i \geq 2.
\end{align}
The parity conditions are always satisfied and the winning probability is given by the probability that the consistency~\eqref{consis condition} holds. 
If $q_1 = 0$, the consistency~\eqref{consis condition} holds if and only if the outcomes are $\{1,1,\cdots,1\}$ or $\{-1,-1,\cdots,-1\}$. Thus the winning probability is given by
\begin{align}\label{win 1}
   \Big\<\prod_{i=1}^m \frac{\1+\cO_{q_i,i}}{2} + \prod_{i=1}^m \frac{\1-\cO_{q_i,i}}{2}\Big\>_\psi.
\end{align} 
On the other hand, if $q_1 = 1$, the consistency \eqref{consis condition} holds if and only if the outcomes are 
$\{1,(-1)^{q_2},\cdots,(-1)^{q_m}\}$ or $\{-1,(-1)^{q_2+1},\cdots,(-1)^{q_m+1}\}$, with the winning probability given by
\begin{align}\label{win 2}
  & \Big\<\frac{\1+\cO_{1,1}}{2}\prod_{i=2}^m \frac{\1+(-1)^{q_i}\cO_{q_i,i}}{2}\notag\\& \hspace{2cm}  + \frac{\1-\cO_{1,1}}{2} \prod_{i=2}^m \frac{\1+(-1)^{q_i+1}\cO_{q_i,i}}{2}\Big\>_\psi
\end{align}
Note that $\cO_{q_i,i}$ are mutually commute and if $r_1 = \pm1$ it holds
\begin{align}   
  \prod_{i=1}^m \frac{1+r_i}{2} + \prod_{i=1}^m \frac{1-r_i}{2} = \prod_{i=2}^m \frac{1+r_1r_i}{2}.
\end{align}
Hence for any given question $\cq$, we can unify the winning probability in Eqs.~\eqref{win 1} and~\eqref{win 2} as 
\begin{align}\label{wining prob for given question}
P_{\cq} = \Big\<\prod_{i=2}^m \frac{\1+(-1)^{q_1q_i}\cO_{q_1,1}\cO_{q_i,i}}{2}\Big\>_\psi.
\end{align}
By straightforward calculation, the average winning probability of our quantum strategy is given by
\begin{align}
\frac{\sum_{\cq}P_{\cq}}{2^m}
= \frac{(1+\cos\alpha)^{m-1} + (1+\sin\alpha)^{m-1}}{2^m}.
\end{align}
Finally, we can choose the optimal parameter $\alpha$ to make the average winning probability as large as possible.

We then proceed to show the optimality of our strategy, i.e., any quantum strategy will not give an average winning probability greater than Eq.~\eqref{quantum value}. In the following we only outline some key steps with detailed proofs delegated to the Supplemental Material. 
Note that the optimal strategy in a game has to be the one satisfying the parity constraint first. If a player replies an answer violating the parity, it is equivalent that he/she gives up and loses this round of the game. This is no better than providing a random answer satisfying the parity and trying his / her luck on the consistency. Thus without loss of generality, we only need to focus on strategies which satisfy the parity conditions.

The main ingredient to show the optimality is to relax the winning conditions of $\MC_m$ and obtain a matching upper bound on the average winning probability.~Specifically, instead of checking the consistency of every single common vertex, we consider checking the product of their assignments. Recall that $\cx_1 \cap \cx_i$ denotes the set of common vertices by the first and the $i$-th player. Denote the product of assignments on $\cx_1 \cap \cx_i$ with respect to $\cP_j$'s $ (j \in \{1,i\})$ answers as 
\begin{align}
  \Pi_{q_1,q_i}(a_j)& := \prod\nolimits_{\cx \in \cx_1 \cap \cx_i} a_j(\cx).
\end{align}
Then it is necessary to win the game that the conditions 
\begin{align}
  \Pi_{q_1,q_i}(a_1) =  \Pi_{q_1,q_i}(a_i), \forall \,i \geq 2 
  \label{eq:relax 1},
\end{align}
hold, and we have the following relaxation of the predicate,
\begin{align}
  V(\ca|\cq) \leq \prod\nolimits_{i=2}^m \big[\Pi_{q_1,q_i}(a_1) =  \Pi_{q_1,q_i}(a_i)\big],
  \end{align}
where $[g]$ is the Iverson bracket, i.e., it takes value $1$ if the statement $g$ is true, otherwise it takes value $0$.

Since we change our concerning objects from the common vertices to the intersecting edges, it naturally induces  observables on the intersecting edges as follows, 
\begin{align}
  \forall\, i \geq 2,\ \cO_{q_1,q_i,1}& := \sum\nolimits_{a_1} \Pi_{q_1,q_i}(a_1) M_{q_1,1}^{a_1},\label{observable 1}\\
  \forall\, i \geq 2,\ \cO_{q_1,q_i,i}&:= \sum\nolimits_{a_i} \Pi_{q_1,q_i}(a_i) M_{q_i,i}^{a_i},\label{observable 2}
\end{align}
where $\cO_{q_1,q_i,1}$ and $\cO_{q_1,q_i,i}$ represent the first and the $i$-th player's observables respectively.

For any quantum strategy satisfying the parity conditions and any given question $\cq$, we will find a similar result as in Eq.~\eqref{wining prob for given question} that
  \begin{align}
  P_{\cq} \leq \Big\<\prod_{i=2}^m \frac{\1+\cO_{q_1,q_i,1}\cO_{q_1,q_i,i}}{2}\Big\>_\psi,
\end{align}
which implies that the average winning probability satisfies
\begin{align}
  \frac{\sum_{\cq}P_{\cq}}{2^m} \leq \frac{\max_i\big\<(\1+\cS_{i})^{m-1}+(\1+\cT_i)^{m-1}\big\>_\psi}{2^m},\label{tmp 3}
\end{align}

\vspace{-0.6cm}
\begin{align}
\text{with} \hspace{0.5cm} \cS_i &:= (\cO_{0,0,1}(\cO_{0,0,i}+\cO_{0,1,i}))/2,\\
\cT_i&:= (\cO_{1,0,1}(\cO_{0,0,i}-\cO_{0,1,i}))/2.
\end{align}
Since $\cO_{q_1,q_i,1}^2 = \cO_{q_1,q_i,i}^2 = \1$, we have $\cS_i^2 + \cT_i^2 = \1, \forall \, i\geq 2$. Combining with Lemma~\ref{tech lemma} below, we have the desired result.

\begin{lemma}\label{tech lemma}
  For any Hermitian operators $\cS$ and $\cT$ satisfying $\cS^2 + \cT^2 = \1$ and any pure state $\ket{\psi}$, $m \geq 2$, it holds that
  \begin{align}
    & \big\<(\1+\cS)^{m-1} + (\1+\cT)^{m-1}\big\>_\psi \notag \\
    & \hspace{1.1cm}\leq \max_{\theta}((1+\cos\theta)^{m-1} + (1+\sin \theta)^{m-1}).
    \label{general CHSH inequality}
  \end{align}
\end{lemma}
We present the proof of this Lemma in the Supplemental Material.
It is worth mentioning that Eq.~\eqref{general CHSH inequality} will reduce to the CHSH inequality for $m = 2$. This implies that Eq.~\eqref{general CHSH inequality} itself is a generalization of the CHSH inequality.

\vspace{0.2cm}
\textit{Hypercube game with classical strategies.}---~The classical strategy can be chosen that every player assigns value $1$ to all the vertices except that the first player assigns $-1$ to the vertex $\cx_0 = (1,\cdots ,1)$ if $q_1 = 1$. Then the players can win the game for the cases when $q_1 = 0$ or $\cq = (1,0,\cdots,0)$. Otherwise, they will lose the game due to the inconsistency at the vertex $\cx_0$. The average wining probability of such strategy is given by $1/2+1/2^m$. 

Following similar steps as the converse part of the quantum strategy, the observable $\cO_{q_1,q_i,1}$ and $\cO_{q_1,q_i,i}$ in Eqs.~\eqref{observable 1} and~\eqref{observable 2} will reduce to classical observables taking values from $\{+1,-1\}$. Thus for any $i$, we have $\cS_i = 0$, $\cT_i = \pm 1$ or $\cS_i = \pm1$, $\cT_i = 0$. Due to Eq.~\eqref{tmp 3} we have the average winning probability no greater than $1/2+1/2^m$. This concludes the optimality of the classical strategy.

\vspace{0.2cm}
\textit{Hypercube game with no-signalling strategies.}---~Denote $K$ as a global assignment to all the $2^m$ vertices and its restricted assignment on the set of vertices $\cx_i$ as $K(\cx_i)$. We say the tuplet $(\ca,\cq)$ is consistent with the global assignment $K$ if $K(\cx_i) = a_i, \forall\, i$. Denote $\cK_{\rm sym}:=\{K\ |\ K(x_1,\cdots, x_m) = K(1-x_1,\cdots,x_m), \forall\, \cx\}$ as the set of all the symmetric assignments with respect to the first coordinate. Let $\boldsymbol \cZ$ be the set of all the tuples $(\ca,\cq)$ consistent with some $K\in\cK_{\rm sym}$ and satisfying the parity condition $\prod_{x\in \cx_1} a_1(x) = (-1)^{q_1}$. Note that the first condition can guarantee the parity of the players ($\cP_i, i\geq 2$) and the consistency of all the players.  We define a correlation as follows,
\begin{align}
  P(\ca|\cq):=\begin{cases}
    \frac{1}{2^{2^{m-1}-1}} & \text{if} (\ca,\cq) \in \boldsymbol\cZ,\\[2pt]
    \quad \ \ 0 & \text{otherwise.}
  \end{cases}
\end{align}

The next step is to show that  $P(\ca|\cq)$ is a legitimate no-signalling correlation.
For any given question $\cq$, the answer $a_1$ completely determines the global answer $\ca$, if $(\ca,\cq)\in \cZ$. Since there are $2^{2^{m-1}-1}$ suitable choices of $a_1$, we have $\sum_{\ca} P(\ca|\cq) = 1$ for any $\cq$.

If $(a_2,\cdots,a_m)$ does not satisfy either parity condition, consistency condition or being symmetric with respect to the first coordinate, we have $P(\ca|\cq) = 0, \forall\ a_1,\cq$. Thus $\sum_{a_1} P(\ca|\cq) = 0$ which is independent of $q_1$. Otherwise, there exists a unique $a_1$  such that $(\ca,\cq) \in \cZ$ since the consistency with other players determines all the assignments of $a_1$ except at the vertex $(q_1,1-q_2,\cdots,1-q_m)$. However, the assignment of this vertex is determined by the parity condition. Thus there is exactly one non-zero term in the summation $\sum_{a_1} P(\ca|\cq)$, i.e., $\sum_{a_1} P(\ca|\cq) = \frac{1}{2^{2^{m-1}-1}}$ and being independent of $q_1$. Without loss of generality, we can make a similar argument with other players and show that $\sum_{a_i} P(\ca|\cq)$ is independent of $q_i$. This concludes that $P(\ca|\cq)$ is a well-defined no-signalling correlation.

Finally, note that $(\ca,\cq) \in \cZ$ implies $V(\ca|\cq) = 1$. For any given question $\cq$, we have
\begin{align}
  \sum_{\ca \in \cA} P(\ca|\cq) V(\ca|\cq) = \sum_{\ca \in \cA} P(\ca|\cq) = 1,
\end{align}
which makes the average winning probability to be one.

\vspace{0.2cm}
\textit{Discussions.}---~We introduced a generalization of the well-known CHSH game to a multipartite scenario and obtained its classical, quantum and no-signalling game values. In particular, the quantum advantage decreases as more players involved while the no-signalling correlation is strong enough to assist the players to win the game deterministically. This set a big difference between no-signalling correlation and quantum correlation in the multipartite setting. 
We leave the analysis of the rigidity and parallel repetition of $\MC_m$ in the future study. Since the CHSH game forms a basis for most state-of-the-art device-independent quantum key distribution (DIQKD) protocols (c.f.~\cite{pironio2009device,reichardt2012classical,arnon2018practical}), our generalization in this work may shed lights on the study of multipartite DIQKD in the future.

\vspace{0.2cm}
\begin{acknowledgments}
We would like to thank Laura Man\v{c}inska for helpful discussions. Note that part of the works were done while RD was at the University of Technology Sydney.
\end{acknowledgments}

\bibliographystyle{apsrev4-1}
\bibliography{BibDoi}

\newpage
\onecolumngrid
\begin{center}
\vspace*{\baselineskip}
{\textbf{\large Supplemental Material \\[3pt] Quantum Advantages in Hypercube Game}}\\[1pt] \quad \\
\end{center}

\renewcommand{\theequation}{S\arabic{equation}}
\renewcommand{\thetheorem}{S\arabic{theorem}}
\setcounter{equation}{0}
\setcounter{figure}{0}
\setcounter{table}{0}
\setcounter{section}{0}

\section{Equivalence of CHSH game}
We show that the CHSH game (or $\MC_2$) presented in the main text is equivalent to the one used in many literatures (e.g.~\cite{cleve2004consequences,Brunner2014}). It is usually presented as follows:
\begin{enumerate}
\item The referee randomly and uniformly chooses the question $q_i\in \{0,1\}$, and sends $q_i$ to the player $\cP_i$. 
\item Each player $\cP_i$ needs to answer one bit $a_i\in\{0,1\}$ to the referee.

\item The players win if and only if $a_1 \oplus a_2 = q_1q_2$.
\end{enumerate}

 Due to the parity constraint in the $\MC_2$, each player only needs to produce one bit of the answer and fill in the other one accordingly. Without loss of generality, the first player assigns $\{a_1,(-1)^{q_1}a_1\}$ to vertices $\{(q_1,0),(q_1,1)\}$ respectively and the second player assigns $\{a_2,a_2\}$ to vertices $\{(0,q_2),(1,q_2)\}$ respectively. According to the winning rules of $\MC_2$, they win if and only if $a_1 = a_2$ for $q_1q_2 = 0$ and $a_1 = -a_2$ for $q_1q_2 = 1$. This is equivalent to the condition in the third item above.

\section{Proof for the optimality of quantum strategy}

The main ingredient to show the optimality of our quantum strategy is to relax the winning conditions of $\MC_m$ and obtain a matching upper bound on the average winning probability.~Specifically, instead of checking the consistency of every single common vertex, we consider checking the product of their assignments. Recall that $\cx_1 \cap \cx_i =\{\,\cx \in\{0,1\}^m | \ x_1 = q_1, x_i=q_i\,\}$ denotes the set of common vertices by the first and the $i$-th player. Denote the product of assignments on $\cx_1 \cap \cx_i$ with respect to $\cP_j$'s $ (j \in \{1,i\})$ answers as 
\begin{align}
  \Pi_{q_1,q_i}(a_j)& := \prod_{\cx \in \cx_1 \cap \cx_i} a_j(\cx).
\end{align}
Then the following conditions are necessary to win the game:
\begin{align}
  \Pi_{q_1,q_i}(a_1) =  \Pi_{q_1,q_i}(a_i), \forall \,i \geq 2 \label{s eq:relax 1},
\end{align}
and we have the following relaxation of the predicate,
\begin{align}\label{s eq: relax predicate}
  V(\ca|\cq) \leq \prod_{i=2}^m \big[\Pi_{q_1,q_i}(a_1) =  \Pi_{q_1,q_i}(a_i)\big],
  \end{align}
where $[g]$ is the Iverson bracket, i.e., it takes value $1$ if the statement $g$ is true, otherwise it takes value $0$.
For any quantum strategy satisfying the parity conditions and any given question $\cq$, we have the winning probability that
\begin{align}
  P_{\cq} & = \Big\<\sum_{\ca} V(\ca|\cq) \prod_{i=1}^m M_{q_i,i}^{a_i}\Big\>_\psi\\
  & \leq \Big\<\sum_{\ca} \prod_{i=2}^m \big[\Pi_{q_1,q_i}(a_1) =  \Pi_{q_1,q_i}(a_i)\big] \prod_{i=1}^m M_{q_i,i}^{a_i}\Big\>_\psi\\
  & = \Big\<\sum_{\ca} \prod_{i=2}^m \Big(\big[\Pi_{q_1,q_i}(a_1) =  \Pi_{q_1,q_i}(a_i)\big] M_{q_i,i}^{a_i}M_{q_1,1}^{a_1}\Big)\Big\>_\psi\\
  & = \Big\<\prod_{i=2}^m \sum_{a_1,a_i} \Big(\big[\Pi_{q_1,q_i}(a_1) =  \Pi_{q_1,q_i}(a_i)\big] M_{q_i,i}^{a_i}M_{q_1,1}^{a_1}\Big)\Big\>_\psi.
  \end{align}
The first line follows from the definition. The second line follows from Eq.~\eqref{s eq: relax predicate}. The third line follows from the fact that $M_{q_i,i}^{a_i}$ are projections. In the last line we swap the summation and the product.

Since we change our concerning objects from the common vertices to the intersecting edges, it naturally induces  observables on the intersecting edges as follows, 
\begin{align}
  \forall\, i \geq 2,\ \cO_{q_1,q_i,1}& := \sum\nolimits_{a_1} \Pi_{q_1,q_i}(a_1) M_{q_1,1}^{a_1},\label{s observable 1}\\
  \forall\, i \geq 2,\ \cO_{q_1,q_i,i}&:= \sum\nolimits_{a_i} \Pi_{q_1,q_i}(a_i) M_{q_i,i}^{a_i},\label{s observable 2}
\end{align}
where $\cO_{q_1,q_i,1}$ and $\cO_{q_1,q_i,i}$ represent the first and the $i$-th player's observables respectively. Note that $[x=y] = \frac12 (1+xy)$ for $x,y\in \{+1,-1\}$, we have 
\begin{align}
\big[\Pi_{q_1,q_i}(a_1) =  \Pi_{q_1,q_i}(a_i)\big] = \frac{1}{2}(1+ \Pi_{q_1,q_i}(a_1)  \Pi_{q_1,q_i}(a_i)).\end{align} and the winning probability
\begin{align}
  P_{\cq}& \leq \Big\<\prod_{i=2}^m \Big(\sum_{a_1,a_i}\frac{M_{q_i,i}^{a_i} M_{q_1,1}^{a_1}}{2} + \sum_{a_1,a_i}\frac{\Pi_{q_1,q_i}(a_1) M_{q_1,1}^{a_1} \Pi_{q_1,q_i}(a_i) M_{q_i,i}^{a_i}}{2} \Big)\Big\>_\psi\\
  & = \Big\<\prod_{i=2}^m \Big(\sum_{a_1,a_i}\frac{M_{q_i,i}^{a_i} M_{q_1,1}^{a_1}}{2} + \frac{\cO_{q_1,q_i,1}\cO_{q_1,q_i,i}}{2} \Big)\Big\>_\psi\label{tmp1}\\
  & = \Big\<\prod_{i=2}^m \frac{\1+\cO_{q_1,q_i,1}\cO_{q_1,q_i,i}}{2}\Big\>_\psi.
\end{align}
The second line follows from the definition in Eqs.~\eqref{s observable 1} and~\eqref{s observable 2}. The last line follows from the completeness of quantum measurement $\sum_{a_i} M_{q_i,i}^{a_i} = \1$. Thus the average winning probability is bounded by
\begin{align}
  \frac{1}{2^m}\sum_{\cq}P_{\cq} & \leq \frac{1}{2^m}\sum_{\cq}\Big\<\prod_{i=2}^m \frac{\1+\cO_{q_1,q_i,1}\cO_{q_1,q_i,i}}{2}\Big\>_\psi = \frac{1}{2^m} \sum_{q_1} \Big\<\prod_{i=2}^m \sum_{q_i}\frac{\1+\cO_{q_1,q_i,1}\cO_{q_1,q_i,i}}{2}\Big\>_\psi.
\end{align}
Due to the parity conditions, we can check that $\cO_{q_1,q_i,1} = (-1)^{q_1} \cO_{q_1,1-q_i,1}$ and $\cO_{q_1,q_i,i} = \cO_{1-q_1,q_i,i}$. Then we have 
\begin{align}
  \frac{1}{2^m}\sum_{\cq}P_{\cq} \leq \frac{1}{2^m} \Big\<\prod_{i=2}^m (\1+\cS_i) + \prod_{i=2}^m (\1+\cT_i)\Big\>_\psi,\label{tmp2}
\end{align}
with 
\begin{align}
\cS_i &:= \frac{1}{2}(\cO_{0,0,1}(\cO_{0,0,i}+\cO_{0,1,i})),\\
\cT_i&:= \frac{1}{2}(\cO_{1,0,1}(\cO_{0,0,i}-\cO_{0,1,i})).
\end{align}
Note that $-\1 \leq \cO_{q_1,q_i,1}, \cO_{q_1,q_i,i} \leq \1$ and $\{\cS_i\}$, $\{\cT_i\}$ are mutually commute. Thus $\1+\cS_i$ and $\1+\cT_i$ are all positive operators. According to the geometric mean inequality, it holds
\begin{align}
  \frac{\sum_{\cq}P_{\cq}}{2^m}&\leq \frac{\sum_{i=2}^m\big\<(\1+\cS_i)^{m-1}+(\1+\cT_i)^{m-1}\big\>_\psi}{2^m(m-1)}\\
  & \leq \frac{\max_i\big\<(\1+\cS_{i})^{m-1}+(\1+\cT_i)^{m-1}\big\>_\psi}{2^m}.\label{s tmp 3}
\end{align}
Since $\cO_{q_1,q_i,1}^2 = \cO_{q_1,q_i,i}^2 = \1$, we have $\cS_i^2 + \cT_i^2 = \1,\ \forall\, i$. Combining with Lemma 2, we have the desired result.

\section{Technical lemmas}

\begin{customlemma}{2}
  For any Hermitian operators $\cS$, $\cT$ such that $\cS^2 + \cT^2 = \1$, any pure state $\ket{\psi}$ and $m \in \mathbb N_+$,  it holds that
  \begin{align}
    \big\<(\1+\cS)^{m} + (\1+\cT)^{m}\big\>_\psi
    \leq \max_{\theta}((1+\cos\theta)^m + (1+\sin \theta)^m).
    \label{s CHSH general}
  \end{align}
\end{customlemma} 
\begin{proof}
 Denote $\cS$'s eigenvectors $\{\ket{u_i}\}$ with corresponding eigenvalues $\{\alpha_i\}$ and $\cT$'s eigenvectors $\{\ket{v_i}\}$ with corresponding eigenvalues $\{\beta_i\}$.  
Since $\cS^2 + \cT^2 = \1$, we can check that for any $\beta_i$,
\begin{align}
\sspan_{\beta_j^2=\beta_i^2}\{\ket{v_j}\}=\sspan_{\alpha_j^2=1-\beta_i^2}\{\ket{u_j}\}.
\end{align}
Denote the Hilbert space $\mathcal{H}=\bigoplus_{i=1}^n\mathcal{H}_i$, where $\mathcal{H}_i:=\sspan_{\beta_j^2=\beta_i^2}\{\ket{v_j}\}$. Decompose $\ket{\psi} = \sum_{i} w_i \ket{\psi_i}$ with $\ket{\psi_i} \in \cH_i$. Then $\cS^k\ket{\psi_i}, \cT^k \ket{\psi_i}\in\mathcal{H}_i$ and $\<\cT^k\>_{\psi_i} \leq |\beta_i|^k$, $\<\cS^k\>_{\psi_i}\leq (\sqrt{1-\beta_i^2})^k$.  We have
\begin{align}
\big\<(\1+\cS)^{m} + (\1+\cT)^{m}\big\>_\psi
=\ & \sum_{k=0}^{m}\binom{m}{k} \<\cS^k+\cT^k\>_\psi\\
= &\sum_{k=0}^{m}\binom{m}{k} \sum_{i=1}^n |w_i|^2 \<\cS^k+\cT^k\>_{\psi_i}\\
\leq&\sum_{k=0}^{m}\binom{m}{k}\sum_{i=1}^n |w_i|^2(\sqrt{1-\beta_i^2}^k+\beta_i^k)\\
=&\sum_{i=1}^{n}|w_i|^2 \sum_{k=0}^{m}\binom{m}{k}(\sqrt{1-\beta_i^2}^k+\beta_i^k)\\
=&\sum_{i=1}^{n} |w_i|^2 ((1+\sqrt{1-\beta_i^2})^{m}+(1+\beta_i)^{m})\\
\leq &\max_\theta ((1+\cos\theta)^{m}+(1+\sin\theta)^{m}).
\end{align}
\end{proof}

\begin{remark}
	The CHSH inequality is stated as 
	\begin{align}
		\<A_0 \ox (B_0 + B_1) + A_1\ox (B_0-B_1)\>_\psi \leq 2\sqrt{2},
		\label{CHSH}
	\end{align}
	where $\psi$ is a pure state, $A_0$ and $A_1$ ($B_0$ and $B_1$) are Alice's (Bob's) observables which are Hermitian operators with spectrum $\{-1,+1\}$. Denote $\cS = A_0 \ox (B_0 + B_1)/2$ and $\cT = A_1\ox (B_0-B_1)/2$. We have $\cS^2 + \cT^2 = \1$. Thus Eq.~\eqref{CHSH} is equivalent to $\<\cS + \cT\>_\psi \leq \sqrt{2}$, which is a special case of Eq.~\eqref{s CHSH general}.
\end{remark}

\begin{lemma}
Denote $r(\theta)=(1+\cos\theta)^{m} + (1+\sin \theta)^{m}$. For any $m \geq 1$, it holds that
\begin{align}
2^{m} + 1 + \frac{m}{2^{m+1}}\leq \max_{\theta}r(\theta) \leq 2^{m} + 1 + \frac{8m}{2^{m+1}}.
\end{align}
\end{lemma}
\begin{proof}
We prove the lower bound first. When $m=1$, let $\theta=\frac{\pi}{4}$, then
$r(\theta) = 2+ \sqrt{2} \geq 2^{m} + 1 + \frac{m}{2^{m+1}}.$

When $m\geq2$, let $\sin \frac{\theta}{2}=\frac{1}{2^{m}}$, then
\begin{align}
r(\theta)&=(1+\cos\theta)^{m} + (1+\sin \theta)^{m}\\
& =  2^{m}(1-\sin^2\frac{\theta}{2})^{m} + (1+\sin\theta)^{m}\\
& \geq  2^{m} (1-m\sin^2\frac{\theta}{2}) + 1+m\sin\theta\\
& =  2^{m}(1-\frac{m}{4^{m}})+1+2m\frac{1}{2^{m}}\sqrt{1-\frac{1}{4^{m}}}\\
& = 2^{m}+1+\frac{m}{2^{m+1}} (4\sqrt{1-\frac{1}{4^{m}}}-2)\\
& \geq 2^{m}+1+\frac{m}{2^{m+1}}.
\end{align}

As for the upper bound, it is easy to observe that $\max_{\theta}r(\theta)=r(\theta_1)\le r(0)+\theta_1\cdot r'(0)$ with $\theta_1 \in (0,\frac{\pi}{4})$. Note that
\begin{align}
r'(\theta)
&=-m\sin\theta(1+\cos\theta)^{m-1}+m\cos\theta(1+\sin\theta)^{m-1}\\
& = m\cos\theta(1+\sin\theta)^{m-1} (1-\frac{\sin \theta}{\cos \theta} (\frac{1+\cos\theta}{1+\sin \theta})^{m-1})\\
& \leq m\cos\theta(1+\sin\theta)^{m-1} (1-\sin \theta (2- \frac{1-\cos\theta+2\sin\theta}{1+\sin\theta})^{m-1})\\
& \leq m\cos\theta(1+\sin\theta)^{m-1} (1-\sin\theta\cdot 2^{m-1} (1-(m-1)\frac{\sin\theta+\sin^2\frac{\theta}{2}}{1+\sin\theta})).
\end{align}
For $\sin\theta = \frac{1}{2^{m-2}}$ and $m\geq 5$, we have
\begin{align}
  r'(\theta) &\leq m\cos\theta(1+\sin\theta)^{m-1} (1-2(1-(m-1)\frac{\sin\theta+\sin^2\frac{\theta}{2}}{2}))\\
  & = m\cos\theta(1+\sin\theta)^{m-1} (-1+(m-1)(\sin\theta+\sin^2\frac{\theta}{2}))\\
  & < m\cos\theta(1+\sin\theta)^{m-1}  (-1 + (m-1)\frac{2}{2^{m-2}})\\
  & \leq 0.
\end{align}
Thus $\theta_1 \leq \frac{1}{2^{m-2}}$ and $\max_\theta r(\theta) \leq 2^m + 1 + \frac{8m}{2^{m+1}}$.
\end{proof}

\end{document}